\documentclass[a4paper,UKenglish]{article}
 
\usepackage{microtype}
\usepackage{amsmath, amssymb, amsthm}
\usepackage{geometry}
 \usepackage{scrextend}

\title{Towards
  Better Separation between Deterministic and Randomized Query
  Complexity} 

\author{Sagnik Mukhopadhyay\thanks{S. Mukhopadhyay is supported
    by a TCS Fellowship.}
    \\ \bigskip Swagato Sanyal\thanks{S. Sanyal is supported by a DAE fellowship.}
    \\Tata Institute
  of Fundamental Research, Mumbai\\ \texttt{\{swagatos\ ,\ sagnik\_m\} @tcs.tifr.res.in}}

\usepackage{amsmath}
\usepackage{amssymb}
\usepackage{amsthm}
\usepackage{fullpage}
\usepackage{mathpazo}
\usepackage{hyperref}
\usepackage{placeins}

\usepackage{algorithm}
\usepackage{algpseudocode}

\usepackage{pifont}

\newtheorem{theorem}{Theorem}
\newtheorem{corollary}[theorem]{Corollary}
\newtheorem{lemma}[theorem]{Lemma}

\newtheorem{definition}[theorem]{Definition}
\newtheorem{claim}[theorem]{Claim}

\newcommand{\lref}[2][]{\hyperref[#2]{#1~\ref*{#2}}}

\begin{document}
\maketitle
\begin{abstract}
We show that there exists a Boolean function $F$ which observes the following separations among deterministic query complexity $(D(F))$, randomized zero error query complexity $(R_0(F))$
and randomized one-sided error query complexity $(R_1(F))$: $R_1(F) = \widetilde{O}(\sqrt{D(F)})$ and $R_0(F)=\widetilde{O}(D(F))^{3/4}$. This refutes the conjecture made by
Saks and Wigderson that for any Boolean function $f$, $R_0(f)=\Omega({D(f)})^{0.753..}$. This also shows widest separation between $R_1(f)$ and $D(f)$ for any Boolean function. The
function $F$ was defined by G{\"{o}}{\"{o}}s, Pitassi and Watson who studied it for showing a separation between deterministic decision tree complexity
and unambiguous non-deterministic decision tree complexity. Independently of us, Ambainis et al proved that different variants of the function $F$ certify optimal (quadratic) 
separation between $D(f)$ and $R_0(f)$, and polynomial separation between $R_0(f)$ and $R_1(f)$. Viewed as separation results, our results are subsumed by those of Ambainis et al.
However, while the functions considerd in the work of Ambainis et al are different variants of $F$, we work with the original function $F$ itself.
\end{abstract}

\section{Introduction}
\label{intro}
The model of decision trees is one of the simplest models of computation. In this model, an algorithm for computing a Boolean function is given query access to the input.
The algorithm queries different bits of the input, possibly in an adaptive fashion, and eventually outputs a bit. The objective is to minimize the number of queries made.
The amount of computation is generally not the quantity of interest in this model. For a Boolean function $f$, The deterministic query complexity $D(f)$ of $f$ is defined to be the
maximum (over inputs) number of queries
the best deterministic query algorithm for $f$ makes. The bounded-error randomized query complexity $R(f)$ of $f$ is defined to be the number of queries made on the worst input by the best
randomized query algorithm for $f$
that is correct with high\footnote{By \emph{high} we mean a constant strictly greater than $\frac{1}{2}$} probability on every input. $R_0(f)$, the zero error
randomized query complexity of $f$, is the expected number of queries made on the worst input
by the best randomized algorithm for $f$ that gives correct answer on each input with probability $1$. Finally $R_1(f)$, the one-sided randomized query complexity of $f$,  is the number of queries made on the worst input by the
best algorithm that is correct on every input with high probability, and in addition correct on every $1$-input with probability $1$. We give formal definitions of
these measures in the next section. 

The relations between these query complexity measures have been extensively studied in the literature. That randomization can save more than a constant factor of queries has been known
for a long time. In their $1986$ paper, Saks and Wigderson \cite{DBLP:conf/focs/SaksW86} gave examples of recursive NAND trees and recursive MAJORITY trees, for which
they credited Snir and Ravi Bopanna respectively. In both these functions, the deterministic and randomized zero-error query complexity are polynomially separated. In the same
paper, Saks and Wigderson studied binary uniform NAND trees, and showed that $R_0(F) = \Theta(D(F)^{0.753..})$ where $F$ is the binary uniform NAND tree function. They also conjectured that this is the widest
separation possible between these two measures of complexity for any Boolean function. For the same function, Santha \cite{DBLP:conf/coco/Santha91} showed that $R(F) =
(1-2\epsilon)R_0(F)$ where $\epsilon$ is the error probability. So for this function we have $R(F) = \Theta(D(F)^{0.753\dots})$. It is easy to see that $D(f) \geq
R(f), R_0(f), R_1(f)$. Blum and Impagliazzo \cite{DBLP:conf/focs/BlumI87}, Tardos \cite{DBLP:journals/combinatorica/Tardos89} and Hartmanis and Hemachandra 
\cite{DBLP:conf/coco/HartmanisH87} independently showed that $R_0(f) \geq \sqrt{D(f)}$. Nisan \cite{DBLP:journals/siamcomp/Nisan91} showed that
for any Boolean function $f$, $D(f) \leq 27R(f)^3$ and $D(f) \leq R_1(f)^2$. The biggest gap known so far between $D(f)$ and $R(f)$ for any $f$ is much less than cubic.

During this work, we came to know of the recent work by Ambainis et al \cite{DBLP:journals/corr/AmbainisBBL15}. In this work the authors prove various separation results
between different query complexity measures. Among several other results, the authors prove:
\begin{enumerate}
 \item There exists a function $f$ for which $R_0(f) = \widetilde{O}(\sqrt{D(f)})$. In view of the lower bound, this is the widest separation possible between these two measures.
 This refutes the conjecture by Saks and Wigderson.
 \item There exists a function $f$ for which $R_0(f) = \widetilde{\Omega}(R_1(f)^{3/2})$.
\end{enumerate}

\subsection{Query Models}
\label{SEC:PREL}

\paragraph{Deterministic query complexity.} A deterministic  query algorithm can be thought of as a rooted binary tree where each internal node is labeled with a variable and each
leaf is labeled with  0 or 1. The algorithm starts by querying the variable at the root of the tree and depending on the value it gets it chooses between its left child and
right child and thus goes down the tree recursively. If the value of a variable at any internal node is 0, the algorithm descends to its left child, otherwise it descends to its right
child.
Whenever the algorithm reaches a leaf, it outputs the value of the leaf and terminates. We say that the query algorithm correctly computes $f$ if for any input $x \in \{0,1\}^n$, the
algorithm outputs $f(x)$. The deterministic query complexity of a function $f$ is defined as follows.

\begin{equation}
D(f) = \min_{T} \mathsf{Depth}(T),
\end{equation}
where $T$ ranges over decision trees which correctly computes $f$.

\paragraph{Randomized query complexity.} A randomized query algorithm can be thought of as a distribution over deterministic query algorithms. A randomized
query algorithm can also be viewed as a query algorithm where each node has an additional power of tossing coins . After querying the variable associated
with any internal node of the tree, the algorithm decides which input bit to query depending on the responses to the queries so far (i.e. the current node in the tree) and the value of the
coin tosses while in that node. It is not hard to see that the two definitions are equivalent. We are interested in two different measures of complexity, one where we do not allow the
algorithm to make error and we measure the expected number of queried variables on an input. Let us denote the expected number of queried variables by algorithm $A$ for evaluating
$f$ on input $x$ by $Q(A,x)$. The zero error randomized query complexity of $f$, denoted by $R_0(f)$, is defined as follows:

\begin{equation}
R_0(f) = \min_A \max_x Q(A,x),
\end{equation}
where $A$ ranges over all randomized query algorithms which correctly computes $f$ on every input. It is to be noted that the expectation is taken over the random coin tosses. Another notion of complexity is
randomized bounded error query complexity, where we allow the query algorithm to err on inputs and we look at the maximum number of queries on any input. We say that a randomized
query algorithm $A$ computes $f$ with probability $\delta$ if for every input $x$, $\Pr_R[A(x) \neq f(x)] \leq \delta$. The bounded error randomized query complexity of $f$, denoted
by $R_\delta(f)$, is defined as follows.

\begin{equation}
R_\delta(f) = \min_A \mathsf{Depth}(A_T),
\end{equation}
where $A_T$ denotes the support of the distribution of binary trees associated with $A$ and we take the minimum over those $A$'s which computes $f$ with error probability $\delta$.
The depth of a collection of trees is interpreted as the maximum depth of any tree in that collection. Since the randomized
bounded error query complexity of $f$ for any two constant error values are within a constant multiplicative factor of each other, we drop the subscript $\delta$ whenever convenient,
and call it $R(f)$.\\
A third notion of query complexity is randomized one-sided query complexity. An input $x$ is said to be a $0$-input ($1$-input) of a function $f$ if $f(x)=0$ $(f(x)=1)$. We say that
a
randomized one-sided error query algorithm $A$ computes $f$ with probability $\delta$ if for every $1$-input, $x$, $\Pr_R[A(x) \neq f(x)] =0$ and for every $0$-input
$x$, $\Pr_R[A(x) \neq f(x)] \leq \delta$. The one-sided error randomized query complexity of $f$, denoted by $R_1^\delta$, is defined as follows.
\begin{equation}
R_1 ^\delta= \min_A \mathsf{Depth}(A_T),
\end{equation}
where $A_T$ denotes the support of the distribution of binary trees associated with $A$ and we take the minimum over those $A$'s which computes $f$ with error probability $\delta$. Since the one-sided error
randomized query complexity of $f$ for any two constant error values are within a constant multiplicative factor of each other, we drop the subscript $\delta$ whenever convenient,
and call it $R_1(f)$.

For our zero-error algorithm we will use the following simple fact: For any Boolean function $f$, $R_0(f)=O(\max \{R_1(f), R_1(\overline{f})\})$.
\subsection{Our results}
\label{us}
In this work we prove the following results.
\begin{theorem}
  \label{th1}
  There exists a Boolean function $F$ for which $R_0(F) =  \widetilde{O}(D(F)^{3/4})$.
\end{theorem}
Theorem \lref{th1} refutes the conjecture made by Saks and Wigderson
\cite{DBLP:conf/focs/SaksW86}, though this result does not match the
lower bound of $R_0(f)$ in terms of $D(f)$. As mentioned in the Introduction, Ambainis et al \cite{DBLP:journals/corr/AmbainisBBL15} exhibit a function that certifies quadratic
separation between $R_0(f)$ and $D(f)$, which is the widest possible in view of matching lower bound.
\begin{theorem}
  \label{th2}
  There exists a Boolean function $F$ for which $R_1(F) =  \widetilde{O}(\sqrt{D(F)})$.
\end{theorem}
This separation matches the lower bound, upto logarithmic factors, on $R_1(f)$ in terms of $D(f)$ for any function $f$. However, since $R_0(f) \geq R_1(f)$, the function used by
Ambainis et al \cite{DBLP:journals/corr/AmbainisBBL15} also certifies the same separation. Thus, viewed as separation results, our results are subsumed by those of Ambainis et al
\cite{DBLP:journals/corr/AmbainisBBL15}. \\
The functions $F$ in Theorems \lref{th1} and \lref{th2}
are the same, and was first defined by G{\"{o}}{\"{o}}s et al \cite{DBLP:journals/eccc/GoosP015a} for showing a gap between deterministic decision tree complexity
and unambiguous non-deterministic decision tree complexity. While the functions used by Ambainis et al are different variants of this function, we work with the original
function itself. \\ \\
We define the function $F$ now. The domain of $F$ is $\mathcal{D}=\{0,1\}^{n(1+\lceil \log n \rceil)}$. An input $M \in \mathcal{D}$ to $F$ is viewed as a matrix of
dimension $\sqrt{n} \times \sqrt{n}$. Each cell $M_{i,j}$ of $M$ consists of two parts:
\begin{enumerate}
 \item A \emph{bit-entry} $b_{i,j} \in \{0,1\}$.
 \item A \emph{pointer-entry} $p_{i,j} \in \{0,1\}^{\lceil  \log n \rceil}$. $p_{i,j}$ is either a valid pointer to some other cell of $M$, or is interpreted as
 $\bot$ (null). If $p_{i,j}$ is not a valid pointer to some other cell, we write \textquotedblleft $p_{i,j}=\bot$\textquotedblright.
\end{enumerate}

Now, we define what we call a \emph{valid pointer chain}. Assume that $t=\sqrt{n}$. For an input $M$ to $F$, a sequence $((i_1,j_1),\ldots,(i_t,j_t))$ of indices in $[\sqrt{n}] \times [\sqrt{n}]$ is called a \emph{valid} pointer chain if:
  \begin{enumerate}
  \item $b_{i_1,j_1}=1$;
  \item $b_{i_2,j_2}=\ldots=b_{i_t,j_t}=0$;
  \item $\forall k < i_1, p_{k,j_1}=\bot$;
  \item for $\ell=1,\ldots, t-1, p_{i_\ell,j_\ell}=(i_{\ell+1},j_{\ell+1})$ and $p_{i_t,j_t}$ is $\bot$;
  \end{enumerate}

$F$ evaluates to $1$ on $M$ iff the following is true:

\begin{enumerate}
\item $M$ contains a unique all 1's column $j_1$, i.e., there exists $j_1 \in [\sqrt{n}]$ such that $\forall i \in [\sqrt{n}]$, $b_{i,j_1}=1$.

\item There exists a valid pointer chain $((i_1,j_1),\ldots,(i_t,j_t))$. This means 
that the column $j_1$ has a cell with non-null pointer entry. $(i_1,j_1)$ is the cell on column $j_1$ with minimum
row index whose pointer-entry is non-null. Starting from $p_{i_1,j_1}$, if we follow the
  pointer, the following conditions
  are satisfied: In each step except the last, the cell reached
  by following the pointer-entry of the cell in the previous step, contains a $0$ as
  bit-entry and a non-null pointer as pointer-entry. In the last
  step, the cell contains a zero as bit-entry and a null
  pointer ($\bot$) as pointer-entry. Also, this pointer chain covers all columns of $M$.
\end{enumerate}

By a simple adversarial strategy, G{\"{o}}{\"{o}}s et
al. \cite{DBLP:journals/eccc/GoosP015a} showed that $D(F) =
\widetilde{\Omega}(n)$. Our contribution is to show the following results.

\begin{lemma}
  \label{th3}
For the function $F$ defined above, $R_0(F) = \widetilde{O}(n^{3/4})$.
\end{lemma}

\begin{lemma}
  \label{th4}
For the function $F$ defined above, $R_1(F) = \widetilde{O}(\sqrt{n})$.
\end{lemma}

Clearly, Lemmas \lref{th3} and \lref{th4} imply Theorems \lref{th1}
and \lref{th2} respectively.

\section{Intuition of the Randomized One-sided Error Query Algorithm for $F$}
We show that the randomized one-sided error query complexity of $F$ is $\widetilde{O}(\sqrt{n})$. In this section we provide intuition for our one-sided error algorithm for $F$. Our algorithm errs on one side:
on $0$-inputs it always outputs $0$ and on $1$-inputs it outputs $1$ with high probability. \\
The algorithm attempts to find a $1$-certificate. If it fails to find a $1$-certificate, it outputs $0$. We show that on every $1$ input, with high probability, the algorithm succeeds in
finding a $1$-certificate. The $1$-certificate our algorithm looks for consists of:
\begin{enumerate}
 \item A column $j$, all of whose bit-entries are $1$'s.
 \item All null pointers of column $j$ till its first non-null pointer-entry.
 \item The pointer chain of length $\sqrt{n}$ that starts from the first non-null pointer entry, and in the next $\sqrt{n}-1$ hops, visits all the other columns.
 The bit entries of all the other cells of the pointer chain than the one in this column are $0$.
\end{enumerate}
To find a $1$-certificate, the algorithm tries to find columns with $0$-cells on them, and adds those columns to a set of discarded columns that it maintains. To this end, a first natural attempt is to repeatedly sample a cell randomly from $M$, and if its bit-entry is $0$, try to follow the pointer originating from that cell. Following
the chain, each time we visit a cell with bit-entry $0$, we can discard the column on which the cell lies. We can expect that, with high probability, after sampling $O(\sqrt{n})$ cells, we land up on some cell in the middle portion of the correct pointer chain that is contained in 
the $1$-certificate (we call this the \emph{principle chain}). Then if we follow that pointer we spend $O(\sqrt{n})$ queries, and eliminate a constant fraction of the existing columns. \\
The problem with this approach is possible existence of other long pointer chains, than the principle chain. It may be the case that we land up on one such chain, of $\Omega(\sqrt{n})$
length, which passes entirely through the columns that we have already discarded. Thus we end up spending $\Omega(\sqrt{n})$ queries, but can discard only one column (the one we began from). \\
To bypass this problem, we start by observing that the principle chain passes through every column, and hence in particular through every undiscarded column. Let $N$ be the number of undiscarded column
at some stage of the algorithm. Note that the length of the principle chain is $\sqrt{n}$. Therefore if we start to follow it from a randomly chosen cell on it, we are expected to see
an undiscarded column in roughly another $\sqrt{n}/N$ hops. In view of this, we modify our algorithm as follows: while following a pointer chain, we check if on an average we are seeing one
undiscarded column in every $\sqrt{n}/N$ hops. If this check fails, we abandon following the pointer, sample another random cell from $M$, and continue. Our procedure \textsc{MilestoneTrace}
does this pointer-traversal. We can prove that conditioned on the event that we land up on the principle chain, the above traversal algorithm enables us to eliminate a constant
fraction of the existing undiscarded columns with high probability. We also show that spending about $\sqrt{n}/N$ queries for each column we eliminate is enough for us to get the desired query
complexity bound. \\
After getting hold of the unique all $1$'s column, the final step is to check if all its bit-entries are indeed $1$'s, and if that can be completed into a full $1$-certificate.
That can clearly be done in $\widetilde{O}(\sqrt{n})$ queries. The \textsc{VerifyColumn} procedure does this.

\section{Bounding $R_1(F)$}

\begin{algorithm}
\begin{algorithmic}[1]

\Procedure{MilestoneTrace($M,\mathcal{C},i,j$)}{}
\State Read $b_{i,j}$;
\If{$b_{i,j}=1$}
    \Return; \label{loss}
\EndIf
\State \emph{step}:=$0$;
\State \emph{discard}:=$1$;
\State \emph{current}:=$(i,j)$;
\State \emph{seen}:=$\{j\}$;
\While{\emph{step} $\leq 100 \sqrt{n} \cdot \frac{\mbox{\emph{discard}}}{|\mathcal{C}|}$}
	\State read the pointer-entry of \emph{current};
	\State \emph{step} $\gets$ \emph{step}$+1$;
	\State \emph{current} $\gets$ pointer-entry of \emph{current};
	\If{\emph{current} is $\bot$} \label{null} goto step\lref{updateC}; \EndIf
	\State read bit-entry of \emph{current};
	\If{\emph{current} is on a column $k$ in $\mathcal{C} \setminus \emph{seen}$ and bit-entry of \emph{current} is $0$}
	      \State \emph{seen} $\gets$ \emph{seen} $\cup \{k\}$;
	      \State \emph{discard} $\gets$ \emph{discard}$+1$;
	\EndIf
\EndWhile
\State $\mathcal{C}\gets\mathcal{C} \setminus \emph{seen}$; \label{updateC}
\EndProcedure

\end{algorithmic}
\end{algorithm}

\begin{algorithm}
\begin{algorithmic}[1]

\Procedure{VerifyColumn($M,k$)}{}
\label{ver}
\State Check if all the bit-entries of cells in the $k$-th column of $M$ are $1$; If not, output $0$;
\State \If{All the pointer-entries of cells in the the $k$-th column of $M$ are $\bot$} \State Output $0$;
\EndIf
\If{The pointer chain starting from the first non-null pointer in column $k$ is valid}
    \State Output $1$;
 \Else
	\State Output $0$;  
\EndIf
\EndProcedure

\end{algorithmic}
\end{algorithm}

\begin{algorithm}
\caption{}
\label{randDT0}
\begin{algorithmic}[1]

\State $\mathcal{C} :=$ set of columns in $M$.

\For{$t=1$ to $O(\sqrt{n} \log n)$}
	\If{$|\mathcal{C}| < 100$}
	\State goto step\lref{out1};
	\EndIf \label{sample1} 
	\State Sample a column $j$ from $\mathcal{C}$ uniformly at random; \label{sample1} 
	\State Sample $i \in [\sqrt{n}]$ uniformly at random; \label{sample2} 
	\State \textsc{MilestoneTrace}$(M,\mathcal{C},i,j)$;
	
\EndFor
\If{$|\mathcal{C}| > 100$ or $|\mathcal{C}|=0$} \label {out1} 
	\State Output $0$;
	\Else
	      \State Read all columns in $\mathcal{C}$;
	      \If{There is a column $k$ with all bit-entries equal to $1$}
	      		\State \textsc{VerifyColumn$(M,k)$};
	      \Else \State Output 0;
	      \EndIf
\EndIf	      

\end{algorithmic}
\end{algorithm}

In this section we give the formal description and analysis of our one-sided error query algorithm for $F$: Algorithm\lref{randDT0}. Algorithm\lref{randDT0} uses two procedures: \textsc{VerifyColumn} and \textsc{MilestoneTrace}. As outlined in the previous section,
\textsc{VerifyColumn}, given a column, checks if all its bit-entries are $1$ and whether it can be completed into a $1$-certificate. \textsc{MilestoneTrace} procedure implements
the pointer traversal algorithm that we described in the preceding paragraph. We next describe the \textsc{MilestoneTrace} procedure in a little more detail. We recall from the last section that the algorithm discards columns in course of its execution. We denote the set of undiscarded columns by $\mathcal{C}$.

\textbf{\underline{\textsc{MilestoneTrace} procedure}}

The functions of the variables used are as follows:

\begin{enumerate}
 \item \emph{step}: Contains the number of pointer-entries queried so far. A bit query is always accompanied by a pointer query, unless the bit is $1$ in which case the traversal stops.
 So upto logarithmic factor, the value in \emph{step} gives us the number of bits queried.
 \item \emph{seen}: Set of columns that were undiscarded before the current run of \textsc{MilestoneTrace}, and that have so far been seen and marked for discarding.
 \item \emph{discard}: size of \emph{seen}
 \item \emph{current}: Contains the indices of the cell currently being considered.
\end{enumerate}
The condition in the \emph{while} loop checks if the number of queries spent is not too much larger than $\frac{\sqrt{n}}{|\mathcal{C}|}$ at any point in time. The \emph{if}
condition in line \lref{null} checks if the current pointer-entry is null. If it is null, $\mathcal{C}$ is updated, and control returns to Algorithm\lref{randDT0}.
The condition in line\lref{null} checks if the pointer chain has reached its end. \\ \\
To analyse Algorithm\lref{randDT0}, we need to prove two statements about \textsc{MilestoneTrace}, which we now informally state. Assume that the algorithm is run on a 
$1$-input.
\begin{enumerate}
 \item Conditioned on the event that a cell $(i,j)$ randomly chosen from the columns in $\mathcal{C}$ is on the principle chain, a call to \textsc{MilestoneTrace}$(M, \mathcal{C}, i, j)$ serves to eliminate a constant fraction of surviving
 columns with high probability.
 \item It is enough to ensure that the average number of queries spent for each eliminated column is not too much larger than $\frac{\sqrt{n}}{|\mathcal{C}|}$. Note that
 $|\mathcal{C}|$ is the number of undiscarded columns during the start of the \textsc{MilestoneTrace} procedure. 
\end{enumerate}
In the following subsection, we prove that Algorithm\lref{randDT0} makes $\widetilde{O}(\sqrt{n})$ queries on every input. In the next subsection we prove that Algorithm\lref{randDT0} succeeds with probability $1$ on $0$-inputs and with probability at
least $2/3$ on $1$-inputs. Lemma\lref{th4} follows from Lemma\lref{query}, Corollary\lref{one-sided} and Lemma\ref{success}.

\subsection{Query complexity of Algorithm \lref{randDT0}}
In this subsection we analyse the query complexity of Algorithm\lref{randDT0}. We bound the total number of $b_{i,j}$'s and $p_{i,j}$'s read by the algorithm. Upto logarithmic factors, that is the total number of bits queried. For the rest of this subsection, one query will mean one query to a bit-entry or a pointer-entry of some cell. \\ 
We first analyse the \textsc{MilestoneTrace} procedure. Recall that $\mathcal{C}$ denotes the set of undiscarded columns.
\begin{lemma}
\label{mt}
Let $i,j$ be such that $b_{i,j}=0$. Let $Q$ and $D$ respectively be the number of queries made and number of columns discarded by a call to \textsc{MilestoneTrace}$(M,\mathcal{C},i,j)$. Then,
\[Q \leq D \cdot \frac{200 \sqrt{n}}{|\mathcal{C}|} + 3\]
\end{lemma}
\begin{proof}
We note that the variable \emph{step} contains the number of pointer queries made so far, and the variable \emph{discard} maintains the number of columns marked so far for discarding.
Every time the \emph{while} loop is entered, \emph{step} $\leq 100 \sqrt{n} \cdot \frac{\mbox{\emph{discard}}}{|\mathcal{C}|}$. In each iteration of the \emph{while} loop, step goes
up by $1$. So at any point, \emph{step} $\leq 100 \sqrt{n} \cdot \frac{\mbox{\emph{discard}}}{|\mathcal{C}|} + 1$. The lemma follows by observing that the total number of bit-entries
queried is at most one more than total number of pointer-entries queried.
\end{proof}
We now use Lemma\lref{mt} to bound the total number of queries made by Algorithm\lref{randDT0}.
\begin{lemma}
\label{query}
Algorithm\lref{randDT0} makes $\widetilde{O}(\sqrt{n})$ queries on each input.
\end{lemma}
\begin{proof}
Whenever $b_{i,j}=1$, \textsc{MilestoneTrace}$(M,\mathcal{C},i,j)$ returns after reading $b_{i,j}$. So the total number of queries made by all calls to \textsc{MilestoneTrace}$(M,\mathcal{C},i,j)$ on such inputs is $\widetilde{O}(\sqrt{n})$. \\
After leaving the \emph{while} loop, the total number of queries required to read constantly many columns in $\mathcal{C}$ and to run \textsc{VerifyColumn} is $O(\sqrt{n})$. \\
Since inside the \emph{while} loop all the queries are made inside the \textsc{MilestoneTrace} procedure, it is enough to show that the total number of queries made by all calls to \textsc{MilestoneTrace}$(M,\mathcal{C},i,j)$ on inputs for which $b_{i,j}=0$ is $\widetilde{O}(\sqrt{n})$. \\
Let $t=\widetilde{O}(\sqrt{n})$ be the total number of calls to \textsc{MilestoneTrace} on such inputs, made in the entire run of Algorithm\lref{randDT0}. Let $s_i$ be the value of $|\mathcal{C}|$ when the $i$-th call to \textsc{MilestoneTrace} is made, and let $s_{t+1}$ be the value of $|\mathcal{C}|$ after the execution of the $t$-th call to \textsc{MilestoneTrace} completes . Let $\Delta s_i$ and $\Delta q_i$ respectively be the number of columns discarded and number of queries made in the $i$-th call to \textsc{MilestoneTrace}. Since $\mathcal{C}$ shrinks only when $b_{i,j}=0$, we have $\Delta s_i=s_{i+1}-s_i$ for $i=1\ldots t$. Since $s_1=\sqrt{n}$, we have that for $i=2,\ldots,t$, $s_i=\sqrt{n}-\displaystyle\sum_{j=1}^{i-1} \Delta s_j$.\\
From lemma\lref{mt} we have $\Delta q_i \leq \Delta s_i \cdot \frac{200 \sqrt{n}}{s_i} + 3$ for $i=1,\ldots,t$. Substituting $\sqrt{n}-\sum_{j=1}^{i-1} \Delta s_j$ for $s_i$ when $i>1$, and adding, we have,
\begin{align}
\sum_{i=1}^t \Delta q_i & \leq 200 \sqrt{n} \cdot \sum_{i=1}^t \frac{\Delta s_i}{s_i}  + 3t \nonumber \\
&=200 \sqrt{n} \cdot \left( \frac{\Delta s_1}{\sqrt{n}} + \sum_{i=2}^t \frac{\Delta s_i}{\sqrt{n}-\sum_{j=1}^{i-1} \Delta s_j} \right) + \widetilde{O}(\sqrt{n})\nonumber \\
&\leq 200 \sqrt{n} \cdot \left( \left( \frac{1}{\sqrt{n}} + \frac{1}{\sqrt{n}-1} + \ldots + \frac{1}{\sqrt{n}-\Delta s_1+1}\right) + \right. \nonumber \\
& \qquad \qquad \qquad \left.\left( \frac{1}{\sqrt{n}-\Delta s_1} + \frac{1}{\sqrt{n}-\Delta s_1-1} + \ldots + \frac{1}{\sqrt{n}-\Delta s_1-\Delta s_2+1}\right) + \ldots \right. \nonumber \\
& \qquad \qquad \qquad \left. + \left(\frac{1}{\sqrt{n}-\sum_{j=1}^{t-1}\Delta s_j} + \frac{1}{\sqrt{n}-\sum_{j=1}^{t-1}\Delta s_j-1} + \ldots + 
\frac{1}{\sqrt{n}-\sum_{j=1}^{t-1}\Delta s_j-\Delta s_t +1}\right)\right)+ \widetilde{O}(\sqrt{n}) \nonumber \\
& \leq O(\sqrt{n}) \cdot \left(\sum_{i=1}^{\sqrt{n}} \frac{1}{i}\right) + \widetilde{O}(\sqrt{n})\nonumber \\
& =O(\sqrt{n} \log n) +\widetilde{O}(\sqrt{n})\nonumber \\
&=\widetilde{O}(\sqrt{n}). \nonumber
\end{align}
Hence proved.
\end{proof}

\subsection{Success Probability of Algorithm\lref{randDT0}}
In this section we prove that Algorithm\lref{randDT0} outputs correct answer with probability $1$ on $0$-inputs and with probability at least $2/3$ on $1$-inputs. We start by a proving a
probability statement
(Lemma\lref{probability}) that will help us in the analysis. \\
Let $x_1, \ldots, x_\ell$ be non-negative real numbers and $\sum_{i=1}^\ell x_i=N$. We say that an index $I \in [\ell]$ is \emph{bad} if there exists a non-negative integer
$0 \leq D \leq N-I$
such that
 \[\sum_{i=I}^{I+D} x_i > 100 (D+1) \cdot \frac{N}{\ell}\]
 We say that an index $I$ is \emph{good} if $I$ is not bad.
 \begin{lemma}
  \label{probability}
  Let $I$ be chosen uniformly at random from $[\ell]$. Then,
  \[\mathbb{P}[\mbox{$I$ is good}] > \frac{99}{100}\]
 \end{lemma}
 \begin{proof}
  We show existence of a set $K=\{J_1, \cdots,J_t\}$ of disjoint sub-intervals of $[1,\ell]$ with integer end-points, having the following properties:
  \begin{enumerate}
   \item Every bad index is in some interval $J_i \in K$.
   \item $\forall 1 \leq i \leq t, \sum_{j \in J_i} x_j > 100 |J_i| \cdot \frac{N}{\ell}$.
  \end{enumerate}
  It then follows that the number of bad indices is upper bounded by $\sum_{i \in [t]}|J_i|$ (by property $1$ and disjointness of the intervals). But
  $N \geq \sum_{i \in [t]} \sum_{j \in J_i} x_j >
  100  \cdot \frac{N}{\ell} \sum_{i \in t} |J_i|$ , which gives us that $\sum_{i \in t} |J_i| < \frac{\ell}{100}$. In the above chain of inequalities, the first inequality follows
  from 
  the disjointness of $J_i$'s and the second inequality follows from property $2$. \\
  Now we describe a greedy procedure to obtain such a set $K$ of intervals. Let $j$ be the smallest bad index. Then there exists a $d$ such that
  $\sum_{i \in [j,j+d]} x_i > 100 (d+1) \cdot \frac{N}{\ell}$.
  We include the interval $[j,j+d]$ in $K$. Then let $j'$ be the smallest bad index greater than $j+d$. Then there exists a $d'$ for which
  $\sum_{i \in [j',j+d']} x_i > 100 (d'+1) \cdot \frac{N}{\ell}$. We include
  $[j', j'+d']$ in $K$. We continue in this way till there is no bad index which is not already contained in some interval in $K$. It is easy to verify that the intervals in the set
  $K$ thus formed
  are disjoint, and the set $K$ satisfies properties $1$ and $2$.
  \end{proof}
  Let us begin by showing that algorithm\lref{randDT0} is correct with probability $1$ on $0$ inputs of $F$.

  \begin{claim}
   \label{verify}
   If Procedure \textsc{VerifyColumn} outputs $1$ on inputs $M$ and $k$, then $M$ is a $1$ input of $F$.
  \end{claim}
  \begin{proof}
   \textsc{VerifyColumn} outputs $1$ only if the column $k$ has all its bit-entries equal to $1$, and if the pointer chain starting from the first non-null pointer entry is valid
   (recall the definition of a \emph{valid pointer chain} from Section\lref{us}).
  From the definition of $F$, for such inputs $F$ evaluates to $1$. 
  \end{proof}
  
  \begin{corollary}
   \label{one-sided}
   Let $M$ be a $0$-input of of $F$. Then algorithm\lref{randDT0} outputs $0$ with probability $1$.
  \end{corollary}
  \begin{proof}
   The corollary follows by observing that if algorithm\lref{randDT0} returns $1$, a call to \textsc{VerifyColumn} also returns $1$, and hence from Claim\lref{one-sided} the input
   is a $1$-input of $F$.
  \end{proof}
  
  Let us now turn to $1$-inputs of $F$. Let $M$ be a $1$-input of $F$, that we fix for the rest of this subsection. Without explicit mention, for the rest of the subsection we assume
  that Algorithm\lref{randDT0} is run on $M$. Since $M$ is a $1$-input, by the definition of $F$, there is a column $C$ such that all its bit-entries are $1$, and the pointer chain
  starting from the
  first non-null pointer-entry of $j$ is valid. Call this pointer chain the \emph{principle chain}. Let $(C=c_1,\ldots,c_{\sqrt{n}})$ be the order of columns of $M$ in which the
  pointer chain crosses them. Let $(C=m_1, \ldots, 
  m_{|\mathcal{C}|})$ be the order of the columns of
  $\mathcal{C}$ in which the pointer chain crosses them. Note that the column $C$ always belongs to $\mathcal{C}$, as a column is discarded only if the bit-entry of some cell on it is $0$. Define $X_i$ to be the number of $c_j$'s between $m_i$ and $m_{i+1}$, including $m_i$, if $i < |\mathcal{C}|$, and the number of $c_j$'s after $m_i$, including $m_i$,
  if $i=|\mathcal{C}|$. Clearly $\displaystyle\sum_{i=1}^{|\mathcal{C}|} X_i=\sqrt{n}$.
  
  \begin{lemma}
   \label{hit}
  Let $(i,j)$ be a randomly chosen cell on the restriction of the principle chain to the columns in $\mathcal{C}$ (i.e. $j \in \mathcal{C}$) and let $|\mathcal{C}|=N \geq 100$. Then with
  probability at least $\frac{97}{100}$ over the choice of $(i,j)$, a run of the procedure 
  \textsc{TraceMilestone} on inputs $M,\mathcal{C}, i,j$ shrinks the size of $\mathcal{C}$ to at least $\frac{99N}{100}$.
  \end{lemma}
  
  \begin{proof}
  By applying Lemma\lref{probability} on the sequence $(X_i)_{i=1}^{|\mathcal{C}|}$ described in the paragraph preceding this lemma, except with probability at least 
  $1/100 + 1/100 + 1/|\mathcal{C}| \leq 3/100$, $j$ is a good
  index,  $j<\frac{99N}{100}$ (i.e. the column $j$ has at least $\frac{N}{100}$ columns of $\mathcal{C}$ ahead of it on the principle chain), and $j \neq C$. Since $j \neq C$, the 
  bit-entry of the cell sampled is 
  $0$, and hence procedure \textsc{TraceMilestone} does not return control in step\lref{loss}. In the procedure \textsc{TraceMilestone}, 
  if \emph{current} is on the principle chain, the condition in line\lref{null} cannot be satisfied unless \emph{current} is the last cell on the chain. Now, if the condition in the 
  while loop is violated while \emph{current} is on the principle chain,
  it implies that $j$ is a bad index. Thus with probability at least $1-3/100=97/100$, the procedure does not terminate as long as all the $\frac{N}{100}$ columns ahead of $j$ are 
  not seen. Since all columns in
  $\mathcal{C}$ that are seen are discarded, we have the lemma.
  \end{proof}
  
  Now, let us bound the number of iterations of the \emph{for} loop of algorithm\lref{randDT0} required to shrink $|\mathcal{C}|$ by a factor of $1/100$.
  
  \begin{lemma}
  \label{itno}
  Assume that at a stage of execution of algorithm\lref{randDT0} where the control is in the beginning of the \emph{for} loop, $|\mathcal{C}|=N$. Then except with probability $1/25$, after $10 \sqrt{n}$
  iterations of the \emph{for} loop, $|\mathcal{C}|$ will become at most $99N/100$. 
  \end{lemma}
  \begin{proof}
   The probability that a cell on the principle chain is sampled in steps\lref{sample1} and\lref{sample2} is $\frac{1}{\sqrt{n}}$. So the probability that in none of the 
   $10 \sqrt{n}$ executions of steps\lref{sample1} and\lref{sample2}, a cell on the principle chain is picked is $(1-\frac{1}{\sqrt{n}})^{10 \sqrt{n}} \leq \frac{1}{100}$.
   Conditioned on the event that a cell on the principle chain is sampled, from lemma \lref{hit}, except with probability $3/100$, $|\mathcal{C}|$ reduces by a factor of $1/100$ in
   the following run of \textsc{MilestoneTrace}. Union bounding we have that except with probability $1/100+3/100=1/25$, after $10 \sqrt{n}$  iterations of the \emph{for} loop,
   $|\mathcal{C}| \leq 99N/100$.
  \end{proof}
  
 Let $t$ be the minimum integer such that $\sqrt{n} \cdot \left(\frac{99}{100}\right)^t < 100$. Thus $t=O(\log n)$. For $i=1,\ldots,t$, let the random variable $Y_i$ be equal to the index of the
 first iteration of the  \emph{for} loop of Algorithm\lref{randDT0} after which $|\mathcal{C}| \leq \sqrt{n}.\left(\frac{99}{100} \right)^i$. Let $Z_1=Y_1$ and for $i=2,\ldots,t$
 define $Z_i=Y_i-Y_{i-1}$. From Lemma\lref{itno}, for each $i$ we have $\mathbb{E}[Z_i] \leq 25 \times 10 \sqrt{n}=O(\sqrt{n})$. By linearity of expectation, we have
 $\mathbb{E}[\displaystyle\sum_{i=1}^t Z_i]=O(\sqrt{n} \log n)$. By Markov's inequality, with probability at least $2/3$, $\displaystyle\sum_{i=1}^t Z_i=O(\sqrt{n} \log n)$. Thus,
 if we choose the constant hidden in the number of iterations of the \emph{for} loop of Algorithm\lref{randDT0} large enough, then with probability at least $2/3$, $|\mathcal{C}|$ shrinks to less than 100.
 Then the \textsc{VerifyColumn} procedure reads all the columns in $\mathcal{C}$ and outputs the correct value of $F$. Thus we have proved the following Lemma.
 
  \begin{lemma}
   \label{success}
   With probability at least $2/3$, algorithm \lref{randDT0} outputs $1$ on a $1$-input.
  \end{lemma}

\section{Zero error query complexity of $F$}
\label{ZERO-ERROR}

We first present a randomized query algorithm which satisfies the following: If the algorithm outputs $0$ then the given input is a 0-input (The algorithm actually exhibits a $0$
-certificate) and if the given input is a 0-input, then the algorithm outputs $0$ with high probability. This algorithm makes $\tilde{O}(n^{3/4})$ queries in worst case.
For the randomized zero-error algorithm we run Algorithm \ref{randDT0} and this algorithm one after another. If Algorithm \ref{randDT0} outputs 1 then we stop and output $1$.
Else, if Algorithm \ref{randDT1} says 0, we stop and output $0$. Otherwise, we repeat. By the standard argument of $\mathsf{ZPP}=\mathsf{RP}\cap \mathsf{coRP}$ we get the randomized
zero-error algorithm. Though the query complexity of Algorithm \ref{randDT0} is $\tilde{O}(\sqrt{n})$, we get the zero-error query complexity of $F$ to be $\tilde{O}(n^{3/4})$
because of the query complexity of Algorithm \ref{randDT1}.

Now we define the notion of column covering and column span which we will use next.

\begin{definition}
\label{DEF:COVER}
For two columns $C_i$ and $C_j$ in input matrix $M$, we say $C_j$ is covered by $C_i$ if there is a cell $(k,i)$ in $C_i$ and a sequence
$(\beta_1,\delta_1),\dots,(\beta_t,\delta_t)$ of pairs from $[\sqrt{n}] \times [\sqrt{n}]$ such that:
\begin{enumerate}
\item $b_{k,i}=0$,
\item $\delta_t=j$,
\item for all $\ell \in [t]$, $b_{\beta_\ell,\delta_\ell}=0$ and
\item $p_{k,i} = (\beta_1,\delta_1)$ and for $\ell = 1,\dots,t-1, p_{(\beta_\ell,\delta_\ell)} = (\beta_{\ell+1},\delta_{\ell+1})$.
\end{enumerate}
\end{definition}

\begin{definition}
\label{DEF:SPAN}
For a column $C$, we define $\mathsf{Span}_C$ to be the subset of columns in $M$ which consists of $C$ and any column which is covered by $C$.
\end{definition}

We first give an informal description of the algorithm and then we proceed to formally analyze the algorithm in Section \ref{SEC:ANAL-ZERO}. As mentioned before this is also a
one-sided algorithm, i.e., it errs on one side but it errs on the different side than that of Algorithm \ref{randDT0}. The $0$-certificates it attempts to capture are as follows: 

\begin{enumerate}
\item If each the columns has a cell with bit-entry $0$, then the function evaluates to $0$. Those bit-entries form a $0$-certificate.
If there are many $0$'s in each column, The algorithms may capture such a certificate in
the first phase (\emph{sparsification}).
\item Two columns $C_1$ and $C_2$ in $M$ such that $C_1 \notin \mathsf{Span}_{C_2}$ and $C_2\notin \mathsf{Span}_{C_1}$. Existence of two such columns makes the existence
of a valid pointer chain impossible. This is captured in the second phase of the algorithm.
\item Lastly, if there is column all of whose bit-entries are $1$, which does not have a valid pointer chain, then that is also a $0$-certificate. The algorithm may capture
such a certificate in the last phase.
\end{enumerate}

The algorithm proceeds as follows: The main goal of the algorithm is to eliminate any column where it finds a 0 in any of its cells. First the algorithm filters out the columns with
large number of $0$'s with high probability by random sampling. The algorithm probes $n^{1/4}$ locations at random in each column and if it finds any $0$ in any column, it eliminates
that column. This step is called \textit{sparsification}. After sparsification, we are guaranteed that all the columns have small number of 0's. Now the remaining columns can have
either of the following two characteristics: First, a large number of the columns in existing column set have large span. This implies that if we choose a column randomly from the
existing columns, the column will span a large number of columns (i.e., a constant fraction of existing columns) with high probability and we can eliminate all of them. The algorithm
does this exactly in the \textbf{procedure A} of the second phase. The other case can be where most of the columns have small spans. We can show that if this is the case, then if we
pick two random columns $C_i$ and $C_j$ from the set of existing columns, $C_i$ will not lie in the span of $C_j$ and vice-versa with high probability, certifying that $F$ is $0$.
This case is taken care of in
the \textbf{procedure B} of the second phase of the algorithm.

The algorithm runs \textbf{procedure A} and \textbf{procedure B} one after another for logarithmic number of steps. If at any point of the iteration, the algorithm finds two columns
which are not in span of each other, the algorithm outputs $0$ and terminates. Otherwise, as the \textbf{procedure A} eliminates the number of existing columns by a constant factor in each iteration,
with logarithmic number of iteration, either we completely exhaust the column set, which is again a $0$-certificate, or we are left with a single column. Then the algorithm checks the
remaining column and the validity of the pointer chain if that column is an all $1$'s column and answers accordingly. This captures the third kind of 0-certificate as mentioned before.
In Algorithm\lref{randDT1}, we set $\tau$ to be the least number such that $\sqrt{n} \cdot (\frac{99}{100})^\tau \leq 1$. Clearly $\tau=O(\log n)$.

\FloatBarrier
\begin{algorithm}
\caption{}
\label{randDT1}
\begin{algorithmic}[1]

\State $\mathcal{C} :=$ Set of columns in $M$;
\State $\tau:=$ Least number such that $\sqrt{n} \cdot (\frac{99}{100})^\tau \leq 1$;
\For{each column $C$ in $\mathcal{C}$}\label{sp_start}
\State Sample $T=10 \cdot n^{1/4} \log n$ cells uniformly at random;
\If{any bit-entry of any cell is $0$}
\State $\mathcal{C} \gets \mathcal{C} \setminus \{C\}$;
\EndIf

\EndFor\label{sp_end}
\For{$t=1$ to $\tau$}\label{for2}

	\If{$|\mathcal{C}| \leq 1$}
	  \State goto step\lref{out} 
	\EndIf
	\Repeat\label{rep}
	\Procedure{A}{}
	\State Sample a column $C$ from $\mathcal{C}$ uniformly at random;
	\State Read all entries of all cells of $C$;
	\If{All bit-entries are $1$}
		\State \textsc{VerifyColumn$(M,C)$};
	\EndIf
	\If{Number of $0$ bit-entries in $C > n^{1/4}$}\label{if1}
		\State Output $1$ and abort;
	\EndIf
	\State For each cell on $C$ with bit-entry $0$, trace pointer and compute $\mathsf{Span}_C$;
	\State $\mathcal{C} \gets \mathcal{C} \setminus \mathsf{Span}_C$;
	\EndProcedure
	\Until{$\alpha \log \log n$ times}
	\Procedure{B}{}
	\State Pick two columns $C_1$ and $C_2$ uniformly at random from $\mathcal{C}$;
	\If{All bit-entries of $C_1$ ($C_2$) are $1$}
		\State \textsc{VerifyColumn$(M,C_1)$} (\textsc{VerifyColumn$(M,C_1)$});
	\EndIf
	\If{Number of $0$ bit-entries in $C_1$ or $C_2 > n^{1/4}$}\label{if2}
		\State Output $1$ and abort;
	\EndIf
	\If{$C_2 \notin \mathsf{Span}_{C_1}$ and $C_1 \notin \mathsf{Span}_{C_2}$} 
		\State Output $0$ and abort;
	\EndIf
	\EndProcedure
\EndFor 
\If{$\mathcal{C} = \emptyset$}\label{out}
	\State Output $0$;
\EndIf
\If{$|\mathcal{C}|=1$}
	\State Let $\mathcal{C}=\{C\}$.
	\State \textsc{VerifyColumn$(M,C )$};
\EndIf
\State Output $1$.
\end{algorithmic}
\end{algorithm}

\subsection{Analysis}
\label{SEC:ANAL-ZERO}

Let's first look at the running time of the algorithm

\begin{claim}
The query complexity of Algorithm \ref{randDT1} is $\widetilde{O}(n^{3/4})$ in worst case.
\end{claim}

\begin{proof}
We count the number of bit-entries and pointer-entries of the input matrix the algorithm probes. Upto logarithmic factor, that is asymptotically same as the number of bits queried.\\
The first \emph{for} loop runs for $\sqrt{n}$ iteration and in each iteration samples $T$ cells from a column.
So the number of probes of the first \emph{for} loop is $O(\sqrt{n}\times T) = \widetilde{O}(n^{3/4})$.

In \textbf{procedure A}, the number of probes needed to scan the column and to trace pointer from the column is $O(n^{3/4})$. In \textbf{procedure B}, the algorithm
has to check the span of two columns, which takes $O(n^{3/4})$ probes. The number of iterations of the \emph{for} loop of line \ref{for2} is at most $\tau=O(\log n)$. Hence the total number
of probes made inside the \emph{for} loop is $\widetilde{O}(n^{3/4})$.

Lastly, \textsc{VerifyColumn} takes $O(\sqrt{n})$ probes. So the total number of probes is bounded by $\widetilde{O}(n^{3/4})$. Thus the claim follows.
\end{proof}

The first \emph{for} loop, i.e., line \ref{sp_start} to \ref{sp_end} is called \textit{sparsification}.
We have the following guarantee after sparsification.

\begin{claim}
\label{CLM:SPAR}
After the sparsification, with probability at least $99/100$, every column in $\mathcal{C}$ has at most $n^{1/4}$ cells with bit-entry $0$.
\end{claim}

\begin{proof}
We will bound the probability that all the $T$ probes in a column outputs $1$ conditioned on the fact that the column has more than $n^{1/4}$ $0$'s. A single probe
in such a column outputs $0$ with probability at least $1/n^{1/4}$. Hence all the probes output $1$ with probability $(1-1/n^{1/4})^T \leq  \frac{1}{100\sqrt{n}}$. By union bound, this happens
to some column in $M$ with probability at most $1/100$.
\end{proof}

This implies that except with probability $1/100$, the \emph{if} conditions of lines\ref{if1} and \ref{if2}are never satisfied.

\begin{claim}
\label{CLM:DICHO}
Either of the following is true in each iteration of the \emph{for} loop of line \ref{for2}:
\begin{enumerate}
\item for a random column $C \in \mathcal{C}$, $|\mathsf{Span}_C| > |\mathcal{C}|/100$ with probability at least $1/100$.
\item For two randomly picked columns $C_i$ and $C_j$ in $\mathcal{C}$, with probability at least $24/25$, $C_j \notin \mathsf{Span}_{C_i}$ and $C_i \notin \mathsf{Span}_{C_j}$.
\end{enumerate}
\end{claim}

\begin{proof}
Suppose $(1)$ does not hold. For two random columns $C_i$ and $C_j$, Let $L_{i}$ $(L_{j})$ be the event that $|\mathsf{Span}_{C_i}|$ $(|\mathsf{Span}_{C_j}|)|>|\mathcal{C}|/100$.
Let $E_{i,j}$ $(E_{j,i})$ be the event that $C_j \in \mathsf{Span}_{C_i}$ $(C_i \in \mathsf{Span}_{C_j})$. Thus we have,
\[\mathbb{P}\{E_{i,j}\}=\mathbb{P}\{L_i\} \cdot \mathbb{P}\{E_{i,j} | L_i\} + \mathbb{P}\{\overline{L_i}\} \cdot \mathbb{P}\{E_{i,j} | \overline{L_i}\} \]
\[\leq \mathbb{P}\{L_i\}+\mathbb{P}\{E_{i,j} | \overline{L_i}\}\]
\[\leq \frac{1}{100} + \frac{1}{100}=\frac{1}{50}\]
Similarly $\mathbb{P}\{E_{j,i}\} \leq \frac{1}{50}$. By union bound, $(2)$ is true;
\end{proof}

Now we are ready prove the correctness of the algorithm. 

\begin{claim}
\label{CLM:CORR-0}
Given a $0$-input, Algorithm \ref{randDT1} outputs $0$ with probability at least $19/20$ .
\end{claim}

\begin{proof}
We first note that after the execution of \emph{for} loop in line \ref{sp_start}, except with probability at most $1/100$ there is no column in $\mathcal{C}$ having more than $n^{1/4}$ cells with
bit-entries $0$.

If the algorithm finds a column all of whose bit-entries are $1$, it gives correct output by a run of \textsc{VerifyColumn}.

Next, we note that if in any iteration of the \emph{for} loop (line \ref{for2}), condition $(2)$ of Claim \ref{CLM:DICHO} is satisfied, then we find a $0$-certificate (i.e. a pair of 
columns, none of which lies in the span of the other) with probability at least $24/25$.

Finally, assume that for each iteration of the \emph{for} loop, condition $(2)$ is not satisfied. This implies that for each iteration of the loop, condition $(1)$ is satisfied
(From Claim \ref{CLM:DICHO}).
As we run \textbf{procedure A} $\alpha \log \log n$ times,  with
probability at least $1- (\frac{99}{100})^{\alpha \log \log n} \geq 1-\frac{1}{100 \tau}$ (for appropriate setting of the constant $\alpha$) we land up on a column whose span is at
least
$|\mathcal{C}|/100$ and hence we eliminate $1/100$ fraction of columns in $\mathcal{C}$, in one of the iterations of the inner \emph{repeat} loop (line \ref{rep}). By union bound, the
probability that there is even one bad \emph{repeat} loop where we do not eliminate $|\mathcal{C}|/100$ columns, is at most $1/100$.
 Thus the probability that after the execution of \emph{for} loop is over, $|\mathcal{C}|>1$, is at most $1/100$.
So, the total error probability is bounded by $1/100+\max \{1/25,1/100\}=1/20$ from which the claim follows.
\end{proof}

\begin{claim}
\label{CLM:CORR-1}
Given a $1$-input, Algorithm \ref{randDT1} outputs 1 with probability 1.
\end{claim}

\begin{proof}
The proof of this claim is straight-forward. As mentioned before, Algorithm \ref{randDT1} outputs $0$ only if it finds a $0$-certificate. As there is no $0$-certificate for a $1$-input
, the algorithm outputs $1$.
\end{proof}

Lemma \ref{th3} follows by combining Claim \ref{CLM:CORR-0} and Claim \ref{CLM:CORR-1}.

\textbf{Acknowledgements}: We thank Arkadev Chattopadhyay, Prahladh Harsha and Srikanth Srinivasan for useful discussions.

\bibliographystyle{plainurl}
\bibliography{ref}

\newcommand{\etalchar}[1]{$^{#1}$}
\begin{thebibliography}{ABB{\etalchar{+}}15}

\bibitem[ABB{\etalchar{+}}15]{DBLP:journals/corr/AmbainisBBL15}
Andris Ambainis, Kaspars Balodis, Aleksandrs Belovs, Troy Lee, Miklos Santha,
  and Juris Smotrovs.
\newblock Separations in query complexity based on pointer functions.
\newblock {\em CoRR}, abs/1506.04719, 2015.

\bibitem[BI87]{DBLP:conf/focs/BlumI87}
Manuel Blum and Russell Impagliazzo.
\newblock Generic oracles and oracle classes (extended abstract).
\newblock In {\em 28th Annual Symposium on Foundations of Computer Science, Los
  Angeles, California, USA, 27-29 October 1987}, pages 118--126, 1987.

\bibitem[GPW15]{DBLP:journals/eccc/GoosP015a}
Mika G{\"{o}}{\"{o}}s, Toniann Pitassi, and Thomas Watson.
\newblock Deterministic communication vs. partition number.
\newblock {\em Electronic Colloquium on Computational Complexity {(ECCC)}},
  22:50, 2015.

\bibitem[HH87]{DBLP:conf/coco/HartmanisH87}
Juris Hartmanis and Lane~A. Hemachandra.
\newblock One-way functions, robustness, and the non-isomorphism of np-complete
  sets.
\newblock In {\em Proceedings of the Second Annual Conference on Structure in
  Complexity Theory, Cornell University, Ithaca, New York, USA, June 16-19,
  1987}, 1987.

\bibitem[Nis91]{DBLP:journals/siamcomp/Nisan91}
Noam Nisan.
\newblock {CREW} prams and decision trees.
\newblock {\em {SIAM} J. Comput.}, 20(6):999--1007, 1991.

\bibitem[San91]{DBLP:conf/coco/Santha91}
Miklos Santha.
\newblock On the monte carlo boolean decision tree complexity of read-once
  formulae.
\newblock In {\em Proceedings of the Sixth Annual Structure in Complexity
  Theory Conference, Chicago, Illinois, USA, June 30 - July 3, 1991}, pages
  180--187, 1991.

\bibitem[SW86]{DBLP:conf/focs/SaksW86}
Michael~E. Saks and Avi Wigderson.
\newblock Probabilistic boolean decision trees and the complexity of evaluating
  game trees.
\newblock In {\em 27th Annual Symposium on Foundations of Computer Science,
  Toronto, Canada, 27-29 October 1986}, pages 29--38, 1986.

\bibitem[Tar89]{DBLP:journals/combinatorica/Tardos89}
G{\'{a}}bor Tardos.
\newblock Query complexity, or why is it difficult to seperate {NP}
  \({}^{\mbox{a}}\) cap co {NP}\({}^{\mbox{a}}\) from {P}\({}^{\mbox{a}}\) by
  random oracles a?
\newblock {\em Combinatorica}, 9(4):385--392, 1989.

\end{thebibliography}
\end{document}